\documentclass[11pt]{article}

\usepackage{fullpage}
\usepackage{amssymb}
\usepackage{amsmath}
\usepackage{amsfonts}
\usepackage{amsthm}
\usepackage{url}
\usepackage{graphicx}

\newcommand{\nc}{\newcommand}
\nc{\heading}[1]{\begin{center} \large \bf #1 \end{center}}
\newcommand{\Ryden}{{Ryd\'{e}n}}
\newcommand{\bzero}{\mbox{\boldmath $0$}}
\newcommand{\bone}{\mbox{\boldmath $1$}}

\newcommand{\hD}{\hat{D}}
\newcommand{\hh}{\hat{h}}

\newcommand{\hm}{\hat{m}}
\newcommand{\hphi}{\hat{\phi}}
\newcommand{\hH}{\hat{H}}

\newcommand{\hR}{\hat{R}}

\newcommand{\cI}{\mathcal{I}}
\newcommand{\cJ}{\mathcal{J}}

\newcommand{\oN}{\overline{N}}
\newcommand{\tL}{\tilde{L}}
\newcommand{\tN}{\tilde{N}}
\newcommand{\tR}{\tilde{R}}

\newcommand{\tZ}{\tilde{Z}}
\newcommand{\tX}{\tilde{X}}
\newcommand{\tS}{\tilde{S}}

\newtheorem{prop}{Proposition}

\begin{document}
%\date{}
\title{An EM Algorithm for Continuous-time Bivariate Markov Chains\thanks{This work
 was supported in part by the U.S.\ National Science Foundation under
 Grant CCF-0916568.
 Part of the work in this paper was presented in a preliminary form at the
 45th Conference on Information Science and Systems hosted by The Johns Hopkins
 University,
 Baltimore, MD, March 23-25, 2011.}}
\author{Brian L. Mark and Yariv Ephraim \\
   Dept. of  Electrical and Computer Engineering \\
   George Mason University \\
   Fairfax, VA 22030, U.S.A.}

\maketitle

\begin{abstract}
We study properties and parameter estimation of finite-state homogeneous
continuous-time {\em bivariate} Markov chains.
Only one of the two processes of the bivariate Markov chain is observable.
The general form of the bivariate Markov chain
studied here makes no assumptions on the structure of the generator of the chain, and hence, neither the underlying process nor the observable process is 
necessarily Markov.  The bivariate Markov chain allows
for simultaneous jumps of the underlying and observable processes. Furthermore, 
the inter-arrival time of observed events is phase-type. The bivariate Markov chain 
generalizes the batch Markovian arrival process as well as the Markov modulated 
Markov process.  We develop an expectation-maximization (EM) procedure for estimating 
the generator of a
bivariate Markov chain, and we demonstrate its performance. The procedure does 
not rely on any numerical integration
or sampling scheme of the continuous-time bivariate Markov chain.
The proposed EM algorithm is equally applicable to multivariate Markov chains. 

\end{abstract}

\section{Introduction}
\label{sec-intro}

We consider the problem of estimating the parameter of
a continuous-time finite-state homogeneous {\em bivariate}
Markov chain. Only one of the two processes of the bivariate Markov chain is observable. The
other is commonly referred to as the underlying process.  
We do not restrict the structure of the generator of the
bivariate Markov chain to have any particular form. Thus, simultaneous jumps of the 
observable and underlying processes are 
possible, and neither of these two processes is necessarily
Markov. In \cite{wei:2002}, a continuous-time bivariate Markov chain
was used to model 
delays and congestion in a computer network, and
a parameter estimation algorithm was proposed.  
The model was motivated by the desire to capture correlations 
observed empirically in samples of network delays.  
A continuous-time multivariate Markov chain was used to model
ion channel currents in \cite{ball:1993}.

The bivariate 
Markov chain generalizes commonly used models
such as the batch Markovian arrival process (BMAP) \cite{breuer:2002,klemm:2003},
the Markov modulated Markov process (MMMP) \cite{ephraim:2009},
and the Markov modulated Poisson process (MMPP) \cite{fischer:1992,ryden:1994,ryden:1996,roberts:2006}.
In the BMAP, for example, the generator is an infinite
upper triangular block Toeplitz matrix. In the MMMP and MMPP, the generator is such that no simultaneous
jumps of the underlying and observable processes are allowed. In addition, the underlying processes
in all three examples are homogeneous continuous-time Markov chains.

We develop an expectation-maximization (EM) algorithm
for estimating the parameter of a continuous-time bivariate Markov chain. 
The proposed EM algorithm is equally applicable to multivariate Markov chains.
An EM algorithm for the MMPP was originally developed by \Ryden \cite{ryden:1996}.  
Using a similar approach, EM algorithms
were subsequently developed for the BMAP 
in \cite{breuer:2002,klemm:2003} and the MMMP in \cite{ephraim:2009}.
The EM algorithm developed in the present paper also relies
on \Ryden's approach.  It
consists of closed-form, stable recursions employing scaling and Van Loan's approach 
for computation of
integrals of matrix exponentials \cite{vanloan:1978}, 
along the lines of~\cite{roberts:2006,ephraim:2009}. 

In the parameter estimation algorithm of \cite{wei:2002}, 
the continuous-time bivariate Markov chain is first 
sampled and the transition matrix of the resulting {\em discrete-time}
bivariate Markov chain is estimated using a variant of
the Baum algorithm \cite{baum:1970}. The generator of the continuous-time
bivariate Markov chain is subsequently obtained from the 
transition matrix estimate. As discussed in~\cite{roberts:2008}, this
approach may lead to ambiguous estimates of the generator of the bivariate Markov chain, and in
some cases it will not lead to a valid estimate. Moreover, the approach does not allow structuring
of the generator estimate since it is obtained as a byproduct of the transition matrix estimate. The
EM algorithm developed in this paper estimates the generator of the 
bivariate Markov chain directly from a sample 
path of the continuous-time observable process. This leads to a more accurate computationally efficient
estimator which is free of the above drawbacks.

The remainder of this paper is organized as follows. In Section~\ref{sec:bivariate},
we discuss properties of the continuous-time bivariate Markov chain and develop 
associated likelihood 
functions. In Section~\ref{sec:EM}, we develop the EM algorithm. In 
Section~\ref{sec:numerical}, we discuss the implementation of the EM algorithm and
provide a numerical
example. Concluding remarks are given in Section~\ref{sec:conclusion}.

\section{Continuous-time Bivariate Markov Chain}
\label{sec:bivariate}

Consider a  finite-state homogeneous continuous-time
bivariate Markov chain 
\begin{align}
Z = (X, S) = \{ (X(t), S(t)), ~t \geq 0 \},
\end{align}
defined on a standard probability space,
and assume that it is irreducible.  The
process $S = \{ S(t), t \geq 0 \}$ is the 
underlying process with state space of say $\{ a_1, \ldots, a_r \}$,
and $X = \{ X(t), t \geq 0 \}$ is the observable process 
with state space of say $\{ b_1, \ldots, b_d \}$. The orders $r$ and $d$
are assumed known.
We assume without loss of generality
that $a_i = i$ for $i= 1, \ldots, r$ 
and $b_l = l$ for $l = 1, \ldots, d$. 
The state space of $Z$ is then given by
$\{ 1, \ldots, d \} \times \{ 1, \ldots, r \}$.  
Neither $X$ nor $S$
need be Markov.  Necessary and sufficient conditions for either process to be 
a homogeneous continuous-time Markov chain are given in~\cite[Theorem 3.1]{ball:1993}.
With probability one,
all sample paths of $Z$ are right-continuous step functions with a finite
number of jumps in any finite interval~\cite[Theorem 2.1]{albert:1962}. 

The bivariate Markov chain is parameterized by a generator matrix
\begin{align}
H = \{ h_{ln}(ij),~ l,n = 1,\ldots d; i, j = 1, \ldots r \},
\end{align}
where the set of joint states $ \{ (l, i) \}$ is ordered lexicographically.
With $P(\cdot)$ denoting the probability measure on the given space,  
\begin{align}
 h_{ln}(ij) = \lim_{\epsilon \rightarrow 0} 
 \frac{1}{\epsilon} P( Z(t +  \epsilon )  =  (n, j) \mid Z(t)  =  (l,i) ) ,
\end{align}
for $(l,i) \neq (n,j)$.
The generator matrix can be expressed as a block
matrix $H = \{ H_{ln},~ l, n = 1, \ldots, d \}$,
where $H_{ln} = \{ h_{ln}(ij), ~i, j = 1, \ldots, r \}$
are $r \times r$ matrices.  The number of independent
scalar values that constitute the generator $H$ is
at most $rd(rd-1)$.  Since none of the rows of $H$ is identically zero,
the submatrix $H_{ll}$, $l \in \{ 1, \ldots, d \}$,
is strictly diagonally dominant, i.e.,
\begin{align}
- h_{ll}(ii)  = \sum_{(n,j): (n,j) \neq (l,i)}  h_{ln}(ij) 
 > \sum_{j: j \neq i} h_{ll}(ij),
\label{eq-H_ll-nonsingular}
\end{align}
for all $i = 1, \ldots, r$,
and thus, $H_{ll}$ is nonsingular~\cite[p. 476]{strang:2003}.

Clearly, the observable process $X(t)$ is a deterministic function of the 
bivariate Markov chain $Z(t)$. Conversely, the pair consisting of a 
univariate Markov chain together with a deterministic function of 
that chain is a bivariate Markov chain (see \cite{rudemo:1973}).

\subsection{Density of observable process}
\label{subsec-pathdensity}

\begin{figure}
\centerline{
\includegraphics[scale=1.0, clip=true]{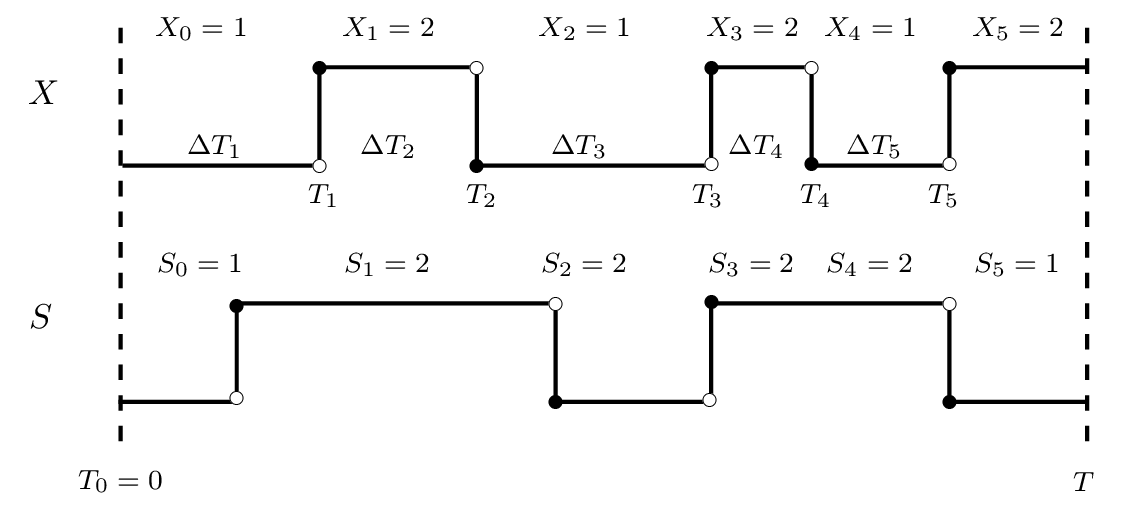}
}
\caption{Example sample path of $Z = (X,S)$.}
 \label{fig-x-s}
\end{figure}

Assume that the observable process $X$ of a bivariate Markov chain $Z = (X,S)$
starts from some state $X_0$
at time $T_0 = 0$ and jumps $N$ times in $[0,T]$
at $0 < T_1 < T_2 < \cdots < T_N \leq T$.  Let $X_k = X(T_k)$
denote the state of $X$ in the interval $[T_k, T_{k+1})$
for $k= 1,2,\ldots, N-1$ and let $X_N$ denote
the state of $X$ in the interval $[T_{N}, T]$.  This convention differs
slightly from that used in \cite{ephraim:2009}, where $X_k \triangleq X(T_{k-1})$,
$k = 1, \ldots, N+1$.  Define $S_k = S(T_k)$ to be the state
of $S$ at the jump time $T_k$ of $X$. Let $Z_k = Z(T_k) = (X_k, S_k)$.
Let $\Delta T_k = T_k - T_{k-1}$ denote the dwell time of $X$
in state $X_{k-1}$ during the interval $[T_{k-1}, T_k)$, $k =  1,\ldots, N$. 
We denote realizations of $X_k$, $S_k$, $Z_k$, $T_k$, and $\Delta T_k$
by $x_k$, $s_k$, $z_k$, $t_k$, and $\Delta t_k$, respectively.
Figure~\ref{fig-x-s} depicts a sample path of a bivariate Markov chain $Z=(X,S)$
for which $N=5$, $r=d=2$, $a_1 = 1$, $a_2 = 2$, $b_1 = 1$, and $b_2 = 2$.
From the figure, we see that the sequence $\{ Z_k \}$ is given by 
\begin{align}
\{ (1,1), (2,2), (1,2), (2,2), (1,2), (2,1) \} .
\end{align}
Note that in Figure~\ref{fig-x-s}, the processes $X$ and $S$
jump simultaneously at times $T_3$ and $T_5$.

The bivariate Markov chain $Z$ is a pure-jump
Markov process and is therefore strong Markov (see \cite[Section 4-1]{wolff:1989})
Using the strong Markov property of $Z$, it follows that
\begin{align}
 P( & Z_{k+1} = z,  \Delta T_{k+1} \leq \tau \mid Z_0, \ldots, Z_k; 
  T_0, \ldots, T_k ) \nonumber \\
&= P(  Z_{k+1} = z,  \Delta T_{k+1} \leq \tau \mid Z_k ) 
\end{align}
for all $z \in \{ 1, \ldots, d \} \times
\{ 1, \ldots, r \}$; $\tau \geq 0$; and $k = 0, 1, 2, \ldots$.
Therefore, $\{ (Z_k, T_k) \}$ is a Markov renewal process (see \cite{cinlar:1975}).
Since the observable process $X$
can be represented 
in an equivalent form as $\{ (X_k, T_k) \}$,
it follows that the density of $X$ in $[ 0, T_N ]$
may be obtained from the product of transition
densities of $\{ (Z_k, T_k ) \}$.  If $T > T_N$, an additional
term is required to obtain the density
of $X$ in $[0,T]$, as will be specified shortly.

Assuming a stationary bivariate Markov chain, the 
transition density of $\{ (Z_k, T_k) \}$ follows
from the density corresponding to  
\begin{align}
P(  Z(\tau) = (n,j),  T_1 \in [\tau, \tau + d\tau) \mid  Z(0) = (l,i) ),
\label{eq-cond-prob-dwell}
\end{align}
for $l \neq n$, which we
denote by $f_{ij}^{ln}(\tau)$.
Let $f^{ln}(\tau) = \{  f_{ij}^{ln}(\tau),~ i, j = 1, \ldots, r  \}$
denote the transition density matrix of $\{ (Z_k, T_k) \}$.
When $\tau$ does not coincide with a jump time
of the observable process, the transition probability 
\begin{align}
\bar{f}^l_{ij}(\tau) = P(  S(\tau)=j , T_1 > \tau \mid X(0)=l, S(0) = i)  \label{eq-fbar_l}
\end{align}
is also required.  
Let $\bar{f}^l(\tau) = \{ \bar{f}^l_{ij}(\tau), ~i,j = 1,\ldots, r \}$
denote the corresponding transition matrix. 
The following proposition
gives explicit forms for $f^{ln}(\tau)$ and $\bar{f}^l(\tau)$.

\begin{prop}\label{prop-fln}
For $\tau \geq 0$,
\begin{align}
	f^{ln}(\tau) = e^{H_{ll} \tau} H_{ln},~~l \neq n , \label{eq-f_ln}
\end{align}
and
\begin{align}
\bar{f}^l(\tau) = e^{H_{ll} \tau} . \label{eq-fbar_l2}
\end{align}
Furthermore, 
\begin{align}
 P( Z_{k+1} = (n,j) \mid Z_k = (l,i) ) = \left[ -H_{ll}^{-1} H_{ln} \right]_{ij} .
 \label{eq-TP-Zk}
\end{align}
\end{prop}

\begin{proof}
Following an argument similar to that given in \cite{freed:1982,ephraim:2009},
the density $f_{ij}^{ln}(t)$ satisfies the following
equation:
\begin{align}
f_{ij}^{ln}(\tau) 
= h_{ln}(ij) e^{h_{ll}(ii) \tau} + 
e^{h_{ll}(ii) \tau} \! \int_0^\tau \! e^{-h_{ll}(ii) t }
  			\sum_{k \neq i} h_{ll}(ik) f_{kj}^{ln} (t) dt  . \label{eq-f-ij}
\end{align}
Differentiating both sides of \eqref{eq-f-ij} with respect to $t$
and simplifying, it follows that
\begin{align}
\frac{d f_{ij}^{ln}(\tau)}{d\tau}
  = h_{ll}(ii) f_{ij}^{ln}(\tau) + \sum_{k \neq i} h_{ll}(ik) f_{kj}^{ln}(\tau) 
  = \sum_k h_{ll}(ik) f_{kj}^{ln}(t\tau) .
\end{align}
Therefore,
\begin{align}
	\frac{d f^{ln}(\tau)}{d\tau} = H_{ll} f^{ln}(\tau)
\end{align}
with initial condition $f^{ln}(0) = H_{ln}$ from \eqref{eq-f-ij}.  Hence,
\eqref{eq-f_ln} follows.  Integrating \eqref{eq-f_ln} from $0$ to $\infty$
gives \eqref{eq-TP-Zk}.

The transition
probability $\bar{f}^l_{ij}(\tau)$ satisfies an equation
similar to \eqref{eq-f-ij} except that the first term on
the right-hand side is replaced by $e^{h_{ll}(ii)\tau} \delta_{ij}$.
This change only affects the initial condition, viz.,
$\bar{f}^l(0) = I$, where $I$ denotes the
identity matrix.  Hence, \eqref{eq-fbar_l2} follows.
\end{proof}

To simplify notation in the sequel, we shall use $P( \cdot )$ to denote not only
a probability measure, but also a density, as appropriate (cf.\ \cite{ryden:1996,ephraim:2009}).  
The exact meaning of expressions involving $P (\cdot)$ should be clear from the
context.  In particular, all null probabilities are to be 
interpreted in the density sense. 
The density of $X$ in $[0, T]$ depends on the initial
state probabilities $\mu_{x_0}(i) = P(X_0 = x_0, S_0 = i)$, $i=1, \ldots, r$.
Let
\begin{align}
\mu_{x_0} = \{ \mu_{x_0}(1) , \mu_{x_0}(2) , \ldots, \mu_{x_0}(r) \} 
\end{align}
denote the initial state distribution.
Using the Markov renewal property of $\{ (Z_k, T_k) \}$, the density
of $X$ in $[0,T]$ can be expressed as
\begin{align}
P( X(t), 0 \leq t \leq T ) = 
 \mu_{x_0}  \left \{ \prod_{k=1}^{N} f^{x_{k-1} x_{k}} (\Delta t_k) \right \}
	\bar{f}^{x_{N}}(T - t_N) \bone  ,
\label{eq-path-density}
\end{align}
where $\bone$ denotes a column vector of all ones.
This expression will be used in Section~\ref{sec:EM}
to develop the EM recursions.

\subsection{Forward-backward recursions}

The density in \eqref{eq-path-density} can be evaluated using forward
and backward recursions. 
The {\em forward density} is defined by the row vector
\begin{align}
L(k) &= \{ P( X(t), 0 \leq t \leq t_k, S_k = i), ~i = 1, \ldots, r   \} , 
\label{eq-L-def}
\end{align}
for $k=0, 1, \ldots, N$.
The forward recursion is given by
\begin{align}
L(0) &= \mu_{x_0} , \nonumber \\
L(k) &= L(k-1) f^{x_{k-1} x_{k}}(\Delta t_k).  \label{eq-forward}
\end{align}
The {\em backward density} is defined by the column vector
\begin{align*}
R(k) &= \{ P ( X(t), t_{k-1} < t \leq T \mid X_{k-1} = x_{k-1}, S_{k-1} = i), ~~ i = 1,\ldots, r \}^\prime,
\end{align*}
for $k = N \! + \! 1, N, \ldots, 1$, where $^\prime$ denote matrix transpose.
The backward recursion is given by
\begin{align}
R(N+1) &=  \bar{f}^{x_{N}}(T - t_N) \bone,  \nonumber \\
R(k) &= f^{x_{k-1} x_{k}}(\Delta t_{k}) R(k+1) .
\label{eq-backward}
\end{align}
From \eqref{eq-path-density}--\eqref{eq-backward}, the density of the observable process 
in $[0, T]$ is given by
\begin{align}
P(X(t), 0 \leq t \leq T) = L(k) R(k+1),~~~~ k= 0, \ldots, N . \label{eq-Pbx}
\end{align}

To ensure numerical stability, it is necessary to scale the
above recursions.  Using an approach similar
to that developed in \cite{roberts:2006},
the scaled forward recursion is given by
\begin{align}
\tL(0) &= \mu_{x_0}, \nonumber \\
\tL(k) &= \frac{\tL(k-1) f^{x_{k-1} x_{k}}(\Delta t_k)}
			   {c_k}, ~k=1, \ldots, N , \label{eq-sforward}
\end{align}
where 
\begin{align}
c_k \triangleq \tL(k-1) f^{x_{k-1} x_{k}}(\Delta t_k) \bone ,~~~k = 1, \ldots, N.
\end{align}
The scaled backward recursion is given by
\begin{align}
\tR(N \! + \! 1) & \! = \! \bar{f}^{x_{N}}(T - t_N) \bone, \nonumber \\
\tR(k) & \! = \!  \frac{f^{x_{k \! - \! 1} x_{k}}(\Delta t_k) \tR(k \! + \! 1)}
			   {c_k}, ~k=1, \ldots, N . \label{eq-sbackward}
\end{align}
Clearly, $\tL(0) = L(0)$ and $\tR(N \! + \! 1) = R(N \! + \! 1)$.
For $k=1, \ldots, N$, one can show straightforwardly that the scaled and unscaled
iterates of the forward and backward recursions are related by
\begin{align}
\tL(k) = \frac{L(k)}{\prod_{m=1}^k c_m} ~\mbox{and}~
\tR(k) = \frac{R(k)}{\prod_{m=k}^{N} c_m}  .
\label{eq-scaled-LR}
\end{align}
From \eqref{eq-scaled-LR} and \eqref{eq-L-def}, one sees that 
for $k = 1, \ldots, N$, the scaled forward vector $\tL(k)$
can be interpreted as the probability distribution of the
underlying process $S$ at time $T_k$
conditioned on the observable sample path
up to and including time $t_k$:
\begin{align}
\tL(k) = \left \{  P( S_k = i \mid X(t), 0 \leq t < t_k  ), i = 1, \ldots, r \right \} .
\end{align}
Thus, the density of $X$ up to and including the $N$th jump can be expressed
as the product of the scaling constants as follows:
\begin{align}
P(X(t), 0 \leq t \leq t_N ) =  L(N) \bone =   
 \left ( \prod_{k=1}^{N} c_k \right ) \tL(N) \bone =  
 \prod_{k=1}^{N} c_k   .  \label{eq-PX-ck}
\end{align}
Therefore, the log-likelihood of the observed sample path is
given by
\begin{align}
 \mathcal{L} = \sum_{k=1}^N \log c_k  .
\label{eq-loglikelihood}
\end{align}

The above forward and backward recursions can be generalized
to apply to any time $t$ between jump times of the 
observable process.  In particular, for
$t \in [t_k, t_{k+1})$, consider the row vector
\begin{align}
\tilde{\ell}(t) &= \left \{
 P( S(t) = i \mid X(\tau), 0 \leq \tau \leq t ),  ~ i = 1, \ldots, r   \right \} .
\label{eq-l-def}
\end{align}
It follows that
\begin{align}
\tilde{\ell}(t) = \frac{L(k) \bar{f}^{x_k}(t - t_k)}{L(k) \bar{f}^{x_k}(t - t_k) \bone} .
\end{align}
A similar result was derived by~\cite{rudemo:1973}.

\subsection{Properties}
\label{ss:dwell_time}

The observable process $X$ is conditionally Markov given $S$,
and vice versa. In contrast to the MMMP and BMAP (see
Section~\ref{ss:models}),
the underlying process $S$ of the bivariate Markov chain $Z$
need not be Markov. 
A necessary and sufficient
condition for $S$ to be a (homogeneous) Markov chain is
that there exists a matrix $Q$ satisfying~\cite[Theorem 3.1]{ball:1993}
\begin{align}
Q = \sum_{n=1}^d H_{ln} \label{eq-S-ctmc}
\end{align}
for $l=1,\ldots, d$.  
In this case, $Q$ is the generator of $S$.  

Various statistics of $X$ and $S$ can be expressed in terms of
the stationary distribution of $Z$.  We denote this
distribution by $\pi = \{ \pi_{nj} \}$, 
where $\pi_{nj} = \lim_{t \rightarrow \infty} P(Z(t) = (n,j) \mid Z(0) = (l,i) )$
and the set of joint states $\{ (n, j) \}$
is ordered lexicographically.
The vector $\pi$ is the unique solution to
the following system:
\begin{align}
 \pi H = \bzero,~~~~ \pi \bone = 1. \label{eq-pi}
\end{align}

The process $\{ Z_k \}$ is a homogeneous discrete-time Markov chain
with transition probabilities given by \eqref{eq-TP-Zk} in
Proposition~\ref{prop-fln}.  Define the $r \times r$ matrix
$A_{ln} = - H_{ll}^{-1} H_{ln}$, for
$\{ l, n: l \neq n \in \{ 1, \ldots, n \}$.  For  $l = n$,
let $A_{ll}$ be a matrix of all zeros.  
The transition matrix of $\{ Z_k \}$ is given by $A = \{ A_{ln} \}$.
Let $D_H = \mbox{diag} \{ H_{ll}, l = 1,\ldots, d \}$.
It follows that $A = - D_H^{-1} H + I$.
Let $\nu = \{ \nu_{nj} \}$ denote the stationary distribution of $\{ Z_k \}$,
where $\nu_{nj} = \lim_{k \rightarrow \infty} P(Z_k = (n,j) \mid Z_0 = (l,i))$.
The vector $\nu$ is the unique solution to the following
system:
\begin{align}
 \nu A = \nu, ~~~~ \nu \bone = 1. \label{eq-nu}
\end{align}
Then $\nu$ can be related to $\pi$ as follows (cf.\ \cite[(6)]{fischer:1992}):
\begin{align}
\nu = \frac{\pi D_H}{\pi D_H \bone } .
\end{align}

The dwell time of the observable process $X$ in a given state
has a phase-type
distribution \cite[Chapter 2]{neuts:1981}.  The phase-type
distribution generalizes mixtures and convolutions
of exponential distributions and can be used to
approximate a large class of dwell time distributions.
To state this property formally, we denote the 
conditional density of the $k$th dwell time $\Delta T_k$
given that $X_{k-1} = l$ by $f_{\Delta T_k \mid X_{k-1}}(\tau \mid l)$.
Define $\alpha_{li} = \lim_{k \rightarrow \infty} P(S_k = i \mid X_k = l)$
and let $\alpha_l = \{ \alpha_{li}, i = 1, \ldots, r \}$.
The conditional probability $\alpha_{li}$ 
can be expressed in terms of $\nu$ as follows:
\begin{align}
\alpha_{li} = \frac{\nu_{li}}{\sum_i \nu_{li}} . \label{eq-alpha-li}
\end{align}
We then have the following proposition, which is proved 
in \ref{sec-app-dwell}.  

\begin{prop} \label{prop-phase}
The stationary conditional distribution of $\Delta T_k$ given 
$X_{k-1} = l$ is phase-type.  In particular,
\begin{align}
\lim_{k \rightarrow \infty} f_{\Delta T_k \mid X_{k-1}} (\tau \mid l) =
	\alpha_l e^{H_{ll} \tau} \beta_l ,
	 \label{eq-pht}
\end{align}
where $\beta_l = \sum_{n:n \neq l} H_{ln} \bone$.
\end{prop}

The phase-type
dwell time of the observable process
may be useful in explicit durational modeling
embedded in a hidden Markov model system (cf.\ \cite{ferguson:1980,sohn:2007,yu:2006}).
It should be clear that the dwell time of the underlying process $S$ 
also has a phase-type distribution.  This may be
useful in certain applications for which the
underlying process has a non-Markovian character.

\subsection{Relation to other models}
\label{ss:models}

In this section, we relate the continuous-time bivariate Markov chain to 
other widely used stochastic models mentioned in Section~\ref{sec-intro}. 

The MMMP (cf.\ \cite{ephraim:2009}) is a bivariate Markov chain 
$(X, S)$ for which the underlying
process $S$ is a homogeneous irreducible Markov chain.
The observable process $X$, conditioned on $S$,
is a nonhomogeneous irreducible Markov chain.
The generator of $S$ satisfies $Q = \sum_{n=1}^d H_{ln}$,
for all $l$, and the submatrices
$\{ H_{ln}, l \neq n \}$ of the bivariate generator
matrix $H$ are all diagonal matrices.  This implies that
with probability one, the underlying and observable
chains do not jump simultaneously.
Conditioned on $S(t)=i$,
the generator of the observable process is given by
$G_i = \{ h_{ln}(ii),~ l, n = 1, \ldots, d \}$.  Thus,
the MMMP may be parameterized by $\{ Q, G_1, \ldots, G_r \}$.
The number of independent scalar values that constitute
the parameter of an MMMP is at most $r(r-1) + rd(d-1)$.

The BMAP (cf.\ \cite{lucantoni:1991,klemm:2003}) is a bivariate Markov chain $(N, S)$
for which the underlying process $S$ is a homogeneous
irreducible finite state Markov chain, and the observable process $N$ is
a counting process with state space $\{ 0, 1, 2, \ldots, ... \}$. The size
of each jump of $N$, i.e., the batch size, lies in $\{ 1, 2, \ldots, d-1 \}$.  
The generator of a BMAP is an upper triangular
block Toeplitz matrix with first row given by
$\{ D_0, D_1, \ldots, D_{d-1}, 0, 0, \ldots \}$.
The generator $Q$ of the underlying chain $S$
satisfies $Q = \sum_{m=0}^{d-1} D_m$.
The BMAP becomes a Markovian arrival process (MAP) when 
$d=1$, i.e., the observable process can jump by at most
one.  For this model, we only have $D_0$ and $D_1$.
The MAP in turn becomes an MMPP (cf.\ \cite{fischer:1992,ryden:1994,ryden:1996,roberts:2006}) 
when $D_1$ is a diagonal matrix with the corresponding
Poisson rates along the main diagonal.
In contrast to the MMMP and MMPP, the observable
and underlying chains of the BMAP and MAP can 
jump simultaneously.

The BMAP can be represented in the framework
of this paper using a finite-state bivariate Markov chain 
$Z = (X, S)$ where $X$ is defined by
$X(t) = (N(t) ~\mbox{mod}~ d) + 1$;
i.e., the observable process $X$ records the modulo-$d$ counts
of the counting process $N$ of the BMAP.  In this case, the generator of $Z$
is given by $H = \{ H_{ln}, ~l, n = 1, \ldots, d \}$, where
\begin{align}
H_{ln} \triangleq D_{l - n ~\mbox{mod}~ d}
\end{align}
and $H$ is a block circulant matrix.
The number of independent scalar values that constitute
the parameter of the BMAP is at most $r^2 d - r$.  

\begin{figure}
\centerline{
\includegraphics[scale=0.65, clip=true]{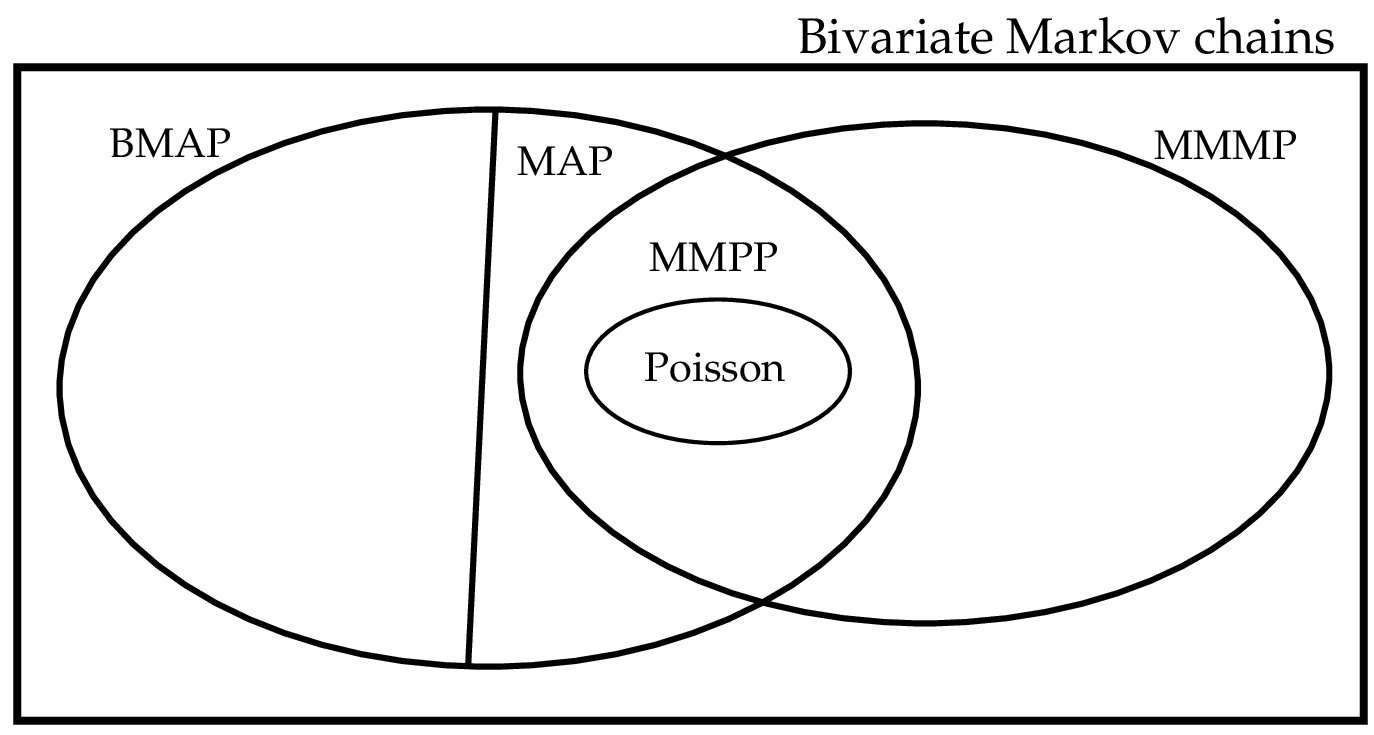}
}
\caption{Relationships among various
 bivariate Markov chains.}
 \label{fig-venn}
\end{figure} 
 
The bivariate Markov chain may also be seen as a hidden Markov process
where $\{ Z_k \}$ plays the role of the 
underlying Markov chain, and the 
observations are continuous random variables
given by $\{ \Delta T_k \}$ \cite{ryden:1996}.  
The conditional density of each observation $\Delta T_k$
depends on both $Z_{k-1}$ and $Z_k$, which
follows from \eqref{eq-f_ln}.
A review of hidden Markov processes may be found in
\cite{ephraim:2002}.  The EM algorithm
developed in 
Section~\ref{sec:EM} is applicable to any of the above particular cases,
as well as multivariate Markov chains (see \cite{ball:1993}).

\section{EM Algorithm}
\label{sec:EM}

In this section, we describe an EM algorithm for
ML estimation of the parameter of a bivariate chain, denoted by $\phi^0$,
given the sample path of the observable process in the interval $[0, T]$.
In the EM approach, a new parameter estimate, say
$\phi_{\iota+1}$, is obtained from a given parameter estimate,
say $\phi_\iota$, as follows:
\begin{align}
\phi_{\iota+1} \! = \! \arg \max_\phi E
\{ \log P ( \{ Z(t), 0 \leq t \leq T \} ; \phi ) \mid 
X(t), 0 \leq t \leq T; \phi_\iota  \} ,
\label{eq-EM}
\end{align}
where the expectation is
taken over $\{ S(t), 0 \leq t \leq T \}$ given 
the observable sample path $\{ X(t), 0 \leq t \leq T \}$.  The maximization is over $\phi$,
which consists of the off-diagonal elements of the
bivariate generator $H$.  
The density of a univariate Markov chain was derived
in \cite{albert:1962}.  A similar approach can be used to
derive the density of the bivariate Markov chain
$\{ Z(t), 0 \leq t \leq T \}$ which is required in \eqref{eq-EM}.
The resulting log-density is expressed in terms of  
the number of jumps $m_{ij}^{ln}$ from each state $(l,i)$
to any other state $(n,j)$ and 
the dwell time $D_i^l$ in each state $(l,i)$.

Let $\varphi_{li}(t) = I( Z(t) = (l,i) )$, where $I(\cdot)$
denotes the indicator function, and let $\#$ denote
set cardinality.  Then,
\begin{align}
m_{ij}^{ln} &= \# \{ t: 0 \! < \!   t \!  \leq \! T, Z(t-) \! = \! (l,i),
	Z(t) \! = \! (n,j) \}  \nonumber \\
 &= \sum_{t \in [0, T]} \varphi_{li}(t-) \varphi_{nj}(t) ,
\label{eq-mijln} \\
D_{i}^l &= \int_0^T \varphi_{li}(t)  dt , \label{eq-Dil} 
\end{align}	
where the sum in \eqref{eq-mijln} is over the jump points of $Z(t)$.
The conditional mean in \eqref{eq-EM} involves the conditional mean estimates
\begin{align}	
\hm_{ij}^{ln} &= E \{ m_{ij}^{ln} \mid X(t), 0 \leq t \leq T  \}  , ~~\mbox{and}~ \label{eq-cme-m} \\
\hD_i^l &= E \{ D_i^l \mid X(t), 0 \leq t \leq T \} ,\label{eq-cme-D} 
\end{align}
where the dependency on $\phi_\iota$ is suppressed.

The maximization in \eqref{eq-EM}
 yields the following intuitive estimate in
the $\iota + 1$st iteration of the EM algorithm \cite{albert:1962}:
\begin{align}
\hh_{ln}(ij) = \frac{\hm_{ij}^{ln}}{\hD_i^l},~~~ (l,i) \neq (n,j).
\end{align}
Next, we develop closed-form expressions for the estimates 
$\hm_{ij}^{ln}$ and $\hD_i^l$.

\subsection{Number of jumps estimate}

The conditional
expectation of $m_{ij}^{ln}$ in \eqref{eq-mijln} is given by
\begin{align}
\hm_{ij}^{ln} &= \sum_{t \in [0,T]} P( Z(t-) = (l,i), Z(t) = (n,j) \mid 
 X(\tau), 0 \leq \tau \leq T ) .
\label{eq-m-sum-jump-Z}
\end{align}
To further evaluate this expression, 
we consider two cases: 1) $l=n$, $i \neq j$; and 2) $l \neq n$.

\vspace{1em}
\noindent
\underline {\em Case 1)}:  ($l=n$, $i \neq j$)
\vspace{0.5em}

In this case, the sum in \eqref{eq-m-sum-jump-Z}
is over jumps of the underlying process~$S$
from $i$ to $j$ while the observable chain~$X$ remains in
state~$l$.
The estimate in \eqref{eq-m-sum-jump-Z} can be written as
a Riemann integral
by partitioning the interval $[0,T]$
into $\oN$ subintervals of length $\Delta$ such that
$\oN \Delta = T$ and then taking the limit
as $\Delta$ approaches zero:
\begin{align}
\hm_{ij}^{ll} 
 &= \lim_{\Delta \rightarrow 0} \sum_{k = 1}^{\oN} \Delta \cdot
 		\frac{P( Z((k-1) \Delta) = (l,i) , Z(k \Delta ) = (l,j) \mid
 		 X(t), 0 \leq \tau \leq T )}{\Delta} \nonumber \\
&= \int_0^T P( Z(t-)  =  (l, i), Z(t) = (l, j) \mid X(\tau), 0 \leq \tau \leq T ) dt .
 \label{eq-numjumps}
\end{align}
In \eqref{eq-numjumps},
$P( Z(t-)=(l,i),  Z(t)=(l, j) \mid X(\tau), 0 \leq \tau \leq T )$
denotes the conditional density given $X(\tau), 0 \leq \tau \leq T$, of a jump of $Z$
from $(l,i)$ to $(l,j)$ at time $t$.  A result similar to \eqref{eq-numjumps} 
was originally stated in~\cite{albert:1962,asmussen:1996,ryden:1996}.
A detailed proof was provided in~\cite{albert:1962}, in
the context of estimating finite-state Markov chains,
and in~\cite{asmussen:1996}, in the 
context of estimating phase-type distributions. The proof was 
adapted in~\cite{ephraim:2009} for estimating MMMPs. 

We have the following proposition, which is stated
for the case when $T = t_N$.
\begin{prop}  \label{prop-njumps-1}
For $k=0, \ldots, N-1$,
define the $2r \times 2r$ matrix
\begin{align}
C_k = \left [
  \begin{array}{cc}
     H_{x_k x_k}	&  H_{x_k x_{k+1}} \tR(k+2) \tL(k) \\
     0				&	  H_{x_k x_k} 
   \end{array}
\right ] .
\end{align}
Let $\cI_k$ be the $r \times r$ upper right block
of the matrix exponential $e^{C_k \Delta t_{k+1}}$, denoted by
\begin{align}
 \cI_k = \left [ e^{C_k \Delta t_{k+1}} \right ]_{12} .
 \label{eq-Ik-comp}
\end{align}
Then,
\begin{align}
\hm_{ij}^{ll} & = 
\left [ H_{ll} \odot \sum_{k:x_k = l} \frac{\cI_k^\prime}{c_{k+1}} \right ]_{ij}
\label{eq-m-ll-ij}
\end{align}
where $\odot$ denotes element-by-element matrix multiplication.
\end{prop}

\begin{proof}
Let $\bone_i$ denote a column vector with a one in the $i$th
element and zeros elsewhere.  Suppose $t \in [t_k, t_{k+1})$
and $x_k = l$.
Applying \eqref{eq-forward}, \eqref{eq-backward}, 
\eqref{eq-PX-ck}
in that order we obtain
\begin{align}
& P(  Z(t-) = (l,i), Z(t) = (l, j) \mid X(\tau), 0 \leq \tau \leq T )  \nonumber \\
&= \frac{P( Z(t-) = (l,i), Z(t) = (l, j), X(\tau), 0 \leq \tau \leq T )}
    {P(X(\tau), 0 \leq \tau \leq T)}	   
    \nonumber \\
&= \frac{\left \{ \mu_{x_0}
 \prod_{m=1}^{k}  f^{x_{m \! - \! 1} x_{m}} (\Delta t_m)  \right \} }
 {P(X(\tau), 0 \leq \tau \leq T)}  
\bar{f}^{x_k}(t \! - \! t_{k}) \bone_i 
  h_{x_k x_k}(ij) \bone_j^\prime f^{x_k x_{k \! + \! 1}} (t_{k \! + \! 1} \! - \! t) 
  \nonumber \\
&~~~~~~ \cdot \left \{ \prod_{m=k+2}^N \! f^{x_{m \! - \! 1} x_{m}}(\Delta t_m) \right \}
 \bone  \nonumber \\
&= 
 \frac{h_{ll}(ij) L(k) }{P(X(\tau), 0 \leq \tau \leq T)}
  \bar{f}^{x_k}(t \! - \! t_{k}) \bone_i \bone_j^\prime
 f^{x_k x_{k \! + \! 1}}  \! (t_{k \! + \! 1}  \! - \! t)  R(k \! + \! 2) \nonumber \\
&= 
 \frac{h_{ll}(ij)  \tL(k)}{c_{k+1}} 
 \bar{f}^{x_k}(t \! - \! t_{k}) \bone_i \bone_j^\prime 
 f^{x_k x_{k \! + \! 1}}  (t_{k \! + \! 1} \! - \! t)  
 \tR(k \! + \! 2) \nonumber \\
&=   \frac{h_{ll}(ij)}{c_{k+1}}
  \left [ \! f^{x_k x_{k \! + \! 1}}  (t_{k \! + \! 1} \! - \! t)   \tR(k \! + \!2) 
\tL(k)  \bar{f}^{x_k}(t \! - \! t_{k})  \right ]_{ji} .
\label{eq-Pijl}
\end{align}
Substituting \eqref{eq-Pijl} into \eqref{eq-numjumps}, it follows that
\begin{align}
\hm_{ij}^{ll} & = 
 \sum_{k:x_k = l}
 \frac{h_{ll}(ij)}{c_{k+1}} 
  \left [  \int_{t_{k}}^{t_{k+1}} \! f^{x_k x_{k+1}} (t_{k+1}  \! - \!  t) 
  \tR(k \! + \! 2) \tL(k) \bar{f}^{x_k}(t \! -  \! t_{k})  dt \right ]_{ji} .
\label{eq-m-ll-ij-int}
\end{align}
Denoting the integral in the above expression by $\cI_k$ and using 
\eqref{eq-f_ln} and \eqref{eq-fbar_l2}, we obtain
\begin{align}
 \cI_k  =  \int_0^{\Delta t_{k \! + \! 1}} e^{H_{x_k x_k} (\Delta t_{k \! + \! 1}  - y)} 
   H_{x_k  x_{k \! + \! 1}} 
	\tR(k \! + \! 2)  \tL(k)  e^{H_{x_k x_k} y} dy . 
	\label{eq-Ik-exp}
\end{align}
The result \eqref{eq-m-ll-ij} now follows from
\eqref{eq-m-ll-ij-int} and \eqref{eq-Ik-exp}.
Following the approach of \cite{roberts:2006,ephraim:2009}, we apply the
result in \cite{vanloan:1978} to evaluate the integral in \eqref{eq-Ik-exp} and
obtain~\eqref{eq-Ik-comp}.
\end{proof}

\vspace{1em}
\noindent
\underline{\em Case 2)}:  ($l \neq n$)
\vspace{0.5em}

In this case, the sum in 
\eqref{eq-m-sum-jump-Z}  is over the jump points of the
observable process~$X$ from state~$l$ to state~$n$,
irrespective of jumps of $S$.  Hence, the conditional mean
of the number of jumps can be written as
\begin{align}
 \hm_{ij}^{ln} &= \mathop{\sum_{k:x_k=l,}}_{x_{k \! + \! 1}=n}
  P( Z(t_k -) = (l,i), Z(t_k) = (n,j) \mid X(\tau), 0 \leq \tau \leq T ) . \label{eq-numjumps2}
\end{align}
We have the following result, which holds for $T \geq t_N$.

\begin{prop} \label{prop-njumps-2}
Let
\begin{align}
\cJ_k = \tR(k \! + \! 2) \tL(k) e^{H_{x_k x_k} \Delta t_k},~~k=0,\ldots, N \! - \! 1.
\label{eq-Jk-comp}
\end{align}
Then for $l \neq n$,
\begin{align}
 \hm_{ij}^{ln} = \left [ H_{ln}
 	  \odot \mathop{\sum_{k:x_k=l,}}_{x_{k+1}=n} \frac{\cJ^\prime_k}{c_{k+1}} 
     \right ]_{ij}, \label{eq-m-ln-ij}
\end{align}
\end{prop}

\begin{proof}
Suppose $k$ is such that $x_k = l$ and $x_{k+1} = n$.
Similarly to Proposition~\ref{prop-njumps-1},
\begin{align}
 & P( Z(t_k -) \! = \!  (l,i), Z(t_k) \! = \! (n,j) \mid X(\tau), 0 \leq \tau \leq T ) \nonumber \\
 &= 
\frac{P(Z(t_k -) \! = \! (l,i), Z(t_k) \! = \! (n,j), X(\tau), 0 \leq \tau \leq T )}
{P(X(\tau), 0 \leq \tau \leq T)} \nonumber \\
&=
\frac{h_{ln}(ij)  \left \{ \!  \mu_{x_0}
  \prod_{m=1}^{k} f^{x_{m \! - \! 1} x_{m}}(\Delta t_m) \! \right \}}
  {P(X(\tau), 0 \leq \tau \leq T)} 
   \bar{f}^{x_{k \! + \! 1}}(\Delta t_{k \! + \!1}) 
 	\bone_i   \bone_j^\prime 
 	\left \{ \! \prod_{m=k+2}^N f^{x_{m \! - \! 1} x_{m}}
 	(\Delta t_m) \! \right \}  \! \bone \nonumber \\
&= \frac{h_{ln}(ij)}{P(X(\tau), 0 \leq \tau \leq T)}
  L(k)  \bar{f}^{x_k}(\Delta t_k) \bone_i  \bone_j^\prime R(k \! + \! 2) \nonumber \\
&= \frac{h_{ln}(ij)}{c_{k+1}}
  \left [ \tR(k \! + \! 2) \tL(k)  \bar{f}^{x_k}(\Delta t_k)  \right ]_{ji} .
 \label{eq-Zlnij}
\end{align}
Substituting \eqref{eq-Zlnij} into \eqref{eq-numjumps2}, 
\begin{align}
 \hm_{ij}^{ln} &= \mathop{\sum_{k:x_k=l,}}_{x_{k+1}=n}
\frac{h_{ln}(ij)}{c_{k+1}} 
  \left [ \tR(k+2) \tL(k)  \bar{f}^{x_k}(\Delta t_k)  \right ]_{ji}  .
 \label{eq-hmij-ln}
\end{align}
The result follows by using 
\eqref{eq-fbar_l2} in \eqref{eq-hmij-ln} and defining
$\cJ_k$ as the bracketed term in that expression. 
\end{proof}

\subsection{Dwell time estimate}

Next, we provide an expression for the dwell time estimate 
$\hD_i^l$.  Taking the conditional expectation in \eqref{eq-Dil},
it follows that
\begin{align}
\hD_i^l = \int_0^T P(Z(t) = (l, i)  \mid X(\tau), 0 \leq \tau \leq T ) dt  .
 \label{eq-dwelltime} 
\end{align}
We have the following result, which is stated for the case
$T = t_N$.
 
\begin{prop}
\begin{align}
\hD_i^l =
\left [ \sum_{k:x_k=l} \frac{\cI_k^\prime}{c_{k+1}} \right ]_{ii} ,
\label{eq-D-i-l}
\end{align}
where $\cI_k$ is given in \eqref{eq-Ik-comp}.
\end{prop}

\begin{proof}
The integrand in \eqref{eq-dwelltime} can be non-zero
only for values of $t$ for which $X(t)$ is in state~$l$.
Hence, it follows that
\begin{align}
\hD_i^l &= \sum_{k:x_k=l}  \int_{t_{k}}^{t_{k+1}}
	 P( Z(t) = (l,i) \mid X(\tau), 0 \leq \tau \leq T ) dt . \label{eq-dwelltime2}
\end{align}
For $t \in [t_{k}, t_{k+1})$ and $x_k = l$, we
have similarly to Proposition~\ref{prop-njumps-1},
\begin{align}
P( &  Z(t)=(l,i) \mid X(\tau), 0 \leq \tau \leq T ) 
= \frac{ P(Z(t) = (l,i), X(\tau), 0 \leq \tau \leq T ) }{P(X(\tau), 0 \leq \tau \leq T)} \nonumber \\
&= \frac{\left \{ \mu_{x_0}\prod_{m=1}^{k} f^{x_{m-1} x_{m}}(\Delta t_m) \right \}}
  {P(X(\tau), 0 \leq \tau \leq T)}
  \cdot
\bar{f}^{x_k}(t \! - \! t_{k}) \bone_i
 \bone_i^\prime    \nonumber \\
&~~~~~~~~ \cdot
 \left \{ 
 f^{x_k x_{k+1}}(t_{k+1} \! - \!  t) \prod_{m=k+2}^{N} f^{x_{m-1} x_{m}}(\Delta t_m) \right \}
  \bone \nonumber \\
&= \frac{L(k)}{P(X(\tau), 0 \leq \tau \leq T)}  \bar{f}^{x_k}(t \! - \! t_{k}) \bone_i
  \bone_i^\prime f^{x_k x_{k \! + \! 1}} \! (t_{k \! + \! 1} \! - \! t) \!  R(k \! + \! 2) \nonumber \\
&= \frac{1}{c_{k+1}} \!  \left [ f^{x_k x_{k \! + \! 1}}(t_{k \! + \! 1} \! - \! t)
	 \tR(k \! + \! 2) \! 
  \tL(k) \! \bar{f}^{x_k}(t \! - \! t_{k}) \! \right ]_{ii} . \label{eq-Zli}
\end{align}
The result follows from substituting \eqref{eq-Zli} into \eqref{eq-dwelltime2}.
\end{proof}

\section{Implementation and Numerical Example}
\label{sec:numerical}

The EM algorithm for continuous-time bivariate Markov chains 
developed in Section~\ref{sec:EM} was
implemented in Python using the SciPy and NumPy libraries.
The matrix exponential function from the SciPy library
is based on a Pad\'{e} approximation, which
has a computational complexity of $O(r^3)$
for an $r \times r$ matrix (see \cite{moler:2003}).
For comparison purposes, the parameter
estimation algorithm based on time-sampling proposed
in \cite{wei:2002} was also implemented in Python.
We refer to this algorithm as the {\em Baum-based
algorithm} for estimating the parameter of
continuous-time bivariate Markov chains.

\subsection{Baum-based Algorithm}

In the Baum-based algorithm described in \cite{wei:2002}, 
the continuous-time bivariate Markov chain $Z$ is  
time-sampled to obtain a discrete-time bivariate Markov chain 
$\tZ = (\tX, \tS) = \{ \tZ_k = (\tX_k, \tS_k) \}$, where
\[
\tZ_k = Z(k \Delta),~~~ k = 0, 1, 2, \ldots,
\]
and $\Delta$ is the sampling interval.  Let $R$ denote the
transition matrix of $\tZ$.  A variant
of the Baum algorithm \cite{baum:1970} is then employed to obtain a maximum 
likelihood estimate, $\hR$, of $R$.  An estimate of the generator
of the continuous-time bivariate Markov chain is obtained
from 
\begin{align}
\hH = \frac{1}{\Delta} \ln( \hR ), \label{eq:inv-P}
\end{align}
where $\ln( \hR )$ denotes the principal branch of the
matrix logarithm of $\hR$, which is given by its
Taylor series expansion
\begin{align}
\ln( \hR ) = \sum_{n=1}^\infty (-1)^{n-1} \frac{(\hR-I)^n}{n} ,
\label{eq-matrix-log}
\end{align}
whenever the series converges.
Existence and uniqueness of a generator $\hH$
corresponding to a transition matrix $\hR$
are not guaranteed (see \cite{israel:2001,roberts:2008}).
In practice, existence and uniqueness
of $\hH$ for a given $\hR$ depend
on the sampling interval $\Delta$ (see \cite{roberts:2008}).
Moreover, if a generator matrix $\hH$ of a certain structure is
desired (e.g., a generator for an MMPP), that structure is
difficult to impose through estimation of $\hR$.

A sufficient condition for the series in \eqref{eq-matrix-log}
to converge is that the diagonal entries of $\hR$
are all greater than 0.5, i.e., $\hR$ is strictly diagonally
dominant (see Theorem~2.2 in \cite{israel:2001} and
Theorem~1 in \cite{wei:2002}).  In this case,
the row sums of $\ln( \hR )$ are guaranteed to be zero,
but some of the off-diagonal elements may possibly
be negative \cite[Theorem 2.1]{israel:2001}.  An 
approximate generator can then be obtained by setting the 
negative off-diagonal entries to zero and adjusting the
diagonal elements such that the row sums of the modified
matrix are zero.  If $\hR$ is not strictly diagonally dominant,
the algorithm in \cite{wei:2002} uses the first term
in the series expansion of \eqref{eq-matrix-log} to
obtain an approximate generator, i.e.,
\begin{align}
\hH = \frac{\hR - I}{\Delta} .
\label{eq-H-approx}
\end{align}

\subsection{Computational and Storage Requirements}
\label{subsec:comp}

The computational requirement of the EM algorithm developed
in Section~\ref{sec:EM} depends linearly on the number of
jumps, $N$, of the observable process.  For each jump
of the observable process,
matrix exponentials
for the transition density matrix $f^{x_k x_{k+1}}(\Delta t_k)$ in
\eqref{eq-forward} and \eqref{eq-backward} and 
for the matrix $\cI_k$ in \eqref{eq-Ik-comp} are
computed.  Computation of
the matrix exponential of an $r \times r$ matrix
requires $O(r^3)$ arithmetic operations (see \cite{moler:2003}).  Thus, the computational
requirement due to computation of matrix exponentials is $O(N r^3)$.
The element-by-element matrix multiplications in 
\eqref{eq-m-ll-ij} and \eqref{eq-m-ln-ij} contribute a computational
requirement of $O(N (r^2 d^2))$.  Therefore, 
the overall computational complexity of the EM algorithm can be stated
as $O(N (r^3 + r^2 d^2))$.
The storage requirement of the EM algorithm is dominated by
the (scaled) forward and backward variables $\tL(k)$ and $\tR(k)$.
Hence, the overall storage required is $O(Nr)$.

By comparison, the computational requirement of the Baum-based
algorithm is $O(\tN r^2 d^2)$, where $\tN = T / \Delta$ is the
number of discrete-time samples.  The storage requirement of
the Baum-based algorithm is $O(\tN r d)$.  Clearly, both the computational and
storage requirement of this algorithm are highly dependent on the
choice of the sampling interval $\Delta$.

\subsection{Numerical Example}
\label{subsec:example}

\begin{table}
\centering
 \begin{tabular}{|c||cc|cc|cc|}
 	\cline{2-7}
	\multicolumn{1}{c||}{ }  &  
	\multicolumn{2}{c|}{$\phi^0$}  & \multicolumn{2}{c|}{$\phi_0$}  & \multicolumn{2}{c|}{$\hphi_{\rm em}$}  \\ \hline \hline
       $H_{11}$   &  -70    & 10     & -120  & 30  &  -77.03   &  14.76   \\
                  &   20    & -55    &   2   & -8  &  8.46    &  -47.95     \\ \hline
       $H_{12}$   &  50     & 10     & 70    & 20  & 51.80    &  10.46     \\ 
                  &  25     & 10     & 5     & 1   &  32.66  &  6.83     \\ \hline
       $H_{21}$   &  50     & 0      & 70    & 0   & 49.50   &  0     \\ 
                  &  0      & 10     & 0     & 1   &  0   &  9.54     \\ \hline
       $H_{22}$   &  -60    & 10     & -100  & 30  &  -59.51 & 10.01     \\ 
                  &   20    & -30    &   2   & -3  &  19.61   &  -29.15     \\ \hline
	\end{tabular}
	\caption{$\phi^0 = {\rm true}$; $\phi_0 = {\rm initial}$; 
			$\hat{\phi}_{\rm em} = \mbox{\rm EM-based estimate}$.}
	\label{tab:em}
\end{table}

\begin{table}
\centering
 \begin{tabular}{|c||cc|cc|cc|cc|}
 	\cline{2-9}
 		\multicolumn{1}{c||}{ }  &  
 		\multicolumn{8}{c|}{$\hphi_{\rm baum}$} \\ \hline
	\multicolumn{1}{|c||}{$\Delta$}  &  
	\multicolumn{2}{c|}{$0.1$}  & 
	\multicolumn{2}{c|}{$0.01$}  & 
	\multicolumn{2}{c|}{$0.005$}  & 
	\multicolumn{2}{c|}{$0.0025$}  \\ \hline \hline
$H_{11}$ &  -10.00 & 0.13  & -77.54 & 10.40  & -78.18 & 15.40  & -80.43 & 15.98   \\
         &  2.35   & -6.08 &  11.30 & -49.21 &   6.16 & -48.20 &  8.73  & -48.53   \\ \hline
$H_{12}$ &  0.42   & 9.44  &  57.34 & 9.81   &  49.29 & 13.48  & 52.13  &  12.32     \\ 
         &  1.28   & 2.45  &  20.63 & 17.29  &  33.49 &  8.55  &  31.44 &  8.36     \\ \hline
$H_{21}$ &  0.00   & 7.92  &  7.10  & 2.91   &  52.04 & 1.18   & 51.92  &  0.45     \\ 
         &  1.85   & 1.39  &  2.19  & 15.36  &  2.11  & 9.74   &  0.97  &  10.63     \\ \hline
$H_{22}$ &  -10.00 & 2.18  & -56.29 & 6.28   & -64.22 & 11.01  & -64.51 &  12.13     \\ 
         &   0.53  & -3.76 &  4.21  & -21.77 &  16.09 & -27.93 &  17.72 &  -29.31     \\ \hline
\end{tabular}
	\caption{Parameter estimates obtained using the Baum-based approach 
		of \cite{wei:2002} with different sampling intervals $\Delta$.}
	\label{tab:baum}
\end{table}

A simple numerical example of estimating the parameter of a continuous-time
bivariate Markov chain using the EM procedure developed in Section~\ref{sec:EM}
is presented in Table~\ref{tab:em}.  For this example, 
the number of underlying states is $r=2$ and the number of observable states is $d=2$.
The generator matrix $H$ is displayed in terms of its block matrix
components $H_{ln}$, which are $2 \times 2$ matrices for
$l, n \in \{1, 2 \}$.  The column labeled $\phi^0$ shows the
true parameter value for the bivariate Markov chain.  Similarly,
the columns labeled $\phi_0$ and $\hphi_{\rm em}$ show, 
respectively, the initial parameter and the EM-based estimate rounded
to two decimal places.
The observed data, generated using the true parameter $\phi^0$,
consisted of $N = 10^4$ observable jumps.
The EM algorithm was terminated when the relative difference
of successive log-likelihood values, evaluated
using \eqref{eq-loglikelihood}, fell below $10^{-7}$.

The bivariate Markov chain parameterized by $\phi^0$ in 
Table~\ref{tab:em} is neither a BMAP nor an MMMP.
Indeed, $H$ is not block circulant as in a
BMAP and $H_{12}$ is not diagonal as in an MMMP.
Moreover, according to \cite[Theorem~3.1]{ball:1993},
the underlying process $S$ is not a homogeneous continuous-time  Markov chain since
$H_{11}+H_{12} \neq H_{21} + H_{22}$ (cf.\ \eqref{eq-S-ctmc}).  

The estimate $\hphi_{\rm em}$ was obtained
after 63 iterations of the EM procedure.  An important property of the EM algorithm is that whenever an
off-diagonal element of 
the generator $H$ is zero in the initial
parameter, the corresponding element in any EM iterate 
remains zero. This can be seen 
easily from Propositions~\ref{prop-njumps-1} 
and \ref{prop-njumps-2}. Thus, if 
structural information about $H$ is known, that 
structure can be incorporated into the initial
parameter estimate and it will be preserved by the EM algorithm
in subsequent iterations.
In the example of Table~\ref{tab:em}, $H_{21}$ is
diagonal in the initial parameter $\phi_0$ and retains
its diagonal structure in the estimate $\hphi_{\rm em}$.
We also see that
the estimate of $H_{21}$, which has the diagonal
structure required for an MMMP, is markedly more accurate than
that of $H_{12}$.  Based on the numerical experience gained from
this and other examples, we can qualitatively say that
estimation of the diagonal elements
of $H_{ln}$ ($l \neq n$) tends to be more accurate and requires
fewer iterations than that of the off-diagonal elements.

For comparison purposes, we have implemented the Baum-based approach 
proposed in \cite{wei:2002} and applied it to the bivariate
Markov chain specified 
in Table~\ref{tab:em} with true parameter $\phi^0$ and
initial parameter estimate $\phi_0$, using the sampling intervals
$\Delta = 0.1$, $0.01$, $0.005$, and $0.0025$.  The corresponding
number of discrete-time samples $\tN$ was 2408, 24077, 48153,
and 96305, respectively.
The algorithm
was terminated when the relative difference
of successive log-likelihood values fell below $10^{-7}$.
The number of iterations required for the four sampling
intervals was $499$, $446$, $105$, and $123$, respectively.
In this example, a generator matrix could be obtained
from the transition matrix estimate using \eqref{eq:inv-P}
for all of the sampling intervals except for $\Delta = 0.1$.
When $\Delta = 0.1$, the generator was obtained using
the approximation~\eqref{eq-H-approx}. 

The results are shown in Table~\ref{tab:baum}.
For all of the sampling intervals, the estimate of $H_{21}$
is not a diagonal matrix, but the accuracy of this estimate
appears to improve as $\Delta$ is decreased.  
The Baum-based estimates of the other block matrices $H_{ln}$
also appear to become closer in value to the EM-based estimate
$\hphi_{\rm em}$ shown in Table~\ref{tab:em} as the sampling
interval decreases.  On the other
hand, as $\Delta$ decreases, the computational requirement
of the Baum-based approach increases proportionally,
as discussed in Section~\ref{subsec:comp}.
In the case $\Delta = 0.1$, the parameter estimate
is far from the true parameter, which is not surprising,
as many jumps of the observable process are missed in the sampling
process.  Indeed, the likelihood of the
final parameter estimate $\hphi_{\rm baum}$
obtained in this case is actually lower
than that of the initial parameter estimate $\phi_0$ given in 
Table~\ref{tab:em}.
This example illustrates not only
the high sensitivity of the final parameter estimate with respect to
the size of the sampling interval, but also that the likelihood
values of the continuous-time bivariate Markov chain
may decrease from one iteration to the next in the Baum-based approach.
In contrast, the EM algorithm generates a sequence of
parameter estimates with nondecreasing likelihood values.  Conditions
for convergence of the sequence of parameter estimates
were given in~\cite{wu:1983}.

\section{Conclusion}
\label{sec:conclusion}

We have studied properties of the
continuous-time bivariate Markov chain and developed
explicit forward-backward recursions for estimating its parameter
based on the EM algorithm. 
The proposed EM algorithm does not require any sampling
scheme or numerical integration.  The bivariate Markov chain
generalizes a large class of stochastic models including the MMMP
and the BMAP, which both generalize the MMPP but
are not equivalent.  In its general form, the bivariate Markov chain
has been used to model ion channel currents (see \cite{ball:1993})
and congestion in computer networks (see \cite{wei:2002}).  
Since the proposed EM procedure preserves the zero values in the
estimates of the generator for the bivariate Markov chain,
it can be applied to estimate the parameter of
special cases, for example,
the MMMP and BMAP, by specifying an initial parameter 
estimate of the appropriate form.

\appendix

\section{Proof of Proposition~\ref{prop-phase}}
\label{sec-app-dwell}

The density corresponding to \eqref{eq-cond-prob-dwell} can be
expressed as
\begin{align}
f_{\Delta T_k, X_k, S_k \mid X_{k-1}, S_{k-1}} (\tau, n, j \mid l, i ) 
= [f^{ln}(\tau)]_{ij} ,~~k \geq 1.
\label{eq-f-TXS}
\end{align}
Summing both sides of \eqref{eq-f-TXS} 
over $n$, for $n \neq l$, and over $j$, applying \eqref{eq-f_ln}, 
and using $\beta_l \triangleq \sum_{n:n \neq l} H_{ln} \bone$,
we obtain
\begin{align}
f_{\Delta T_k \mid X_{k-1}, S_{k-1}} (\tau \mid l, i)
 = \left [ e^{H_{ll} \tau} \sum_{n:n \neq l} H_{ln} \bone \right ]_i 
 = \left [ e^{H_{ll} \tau} \beta_l \right ]_i .
 \label{eq-f-T}
\end{align}
Applying the law of total probability,
\begin{align}
f_{\Delta T_k \mid X_{k-1}} (\tau \mid l)
%&=
%  \sum_i P(S_{k-1} \! = \!  i \mid X_{k-1}  \! = \! l) 
%  f_{\Delta T_k \mid X_{k-1}, S_{k-1}} (t \mid l, i) \nonumber \\
 =  \sum_i P(S_{k - 1}  = i | X_{k  - 1}  =  l) \left [ e^{H_{ll}\tau} \beta_l \right ]_i.
 \label{eq-pht-a}
\end{align}
Taking the limit as $k \rightarrow \infty$, it follows that
\begin{align}
\lim_{k \rightarrow \infty} f_{\Delta T_k \mid X_{k-1}} (\tau \mid l) =
	\alpha_l e^{H_{ll} t\tau} \beta_l .
	 \label{eq-pht-b}
\end{align}

Equation~\eqref{eq-pht-b} has the form of a phase-type distribution
parameterized by $(\alpha_l, H_{ll})$  \cite[Chapter 2]{neuts:1981}. Indeed,
the stationary distribution of $\Delta T_k$ conditioned on $X_{k-1} = l$ 
is equivalent to the distribution of the absorption time of a
Markov chain defined on the state space $\{ 1, 2, \ldots, r+1 \}$  with 
initial distribution given by $\alpha_l$ and generator 
matrix given by
\begin{align}
G = \left [
 \begin{array}{cc}
 	H_{ll} & \beta_l \\
 	\bzero & 0 
 \end{array}
\right ]  ,
\end{align}
where $\bzero$ denotes a row vector of all zeros.
Here, $r+1$ is an absorbing state, while the remaining states, $1, \ldots, r$,
are transient.  

\bibliographystyle{plain}
\bibliography{rmp-bib}

\end{document}